\documentclass[11pt,a4paper]{article}
\usepackage[T1]{fontenc}
\usepackage[utf8]{inputenc}
\usepackage{lmodern,microtype,amsmath,amssymb,amsthm}
\usepackage[margin=26mm,top=24mm,bottom=25mm]{geometry}
\usepackage{xcolor,fancyhdr,needspace}
\usepackage{hyperref}
\definecolor{linkink}{HTML}{184F63}
\hypersetup{colorlinks=true,linkcolor=linkink,citecolor=linkink,urlcolor=linkink,
 pdftitle={Liouvillian and D-finite obstructions for a Klein-Gordon equation with an alpha-attractor-type potential},
 pdfauthor={Benjamin de Zayas Nuñez and Clara Rojas}}
\newcommand{\C}{\mathbb C}
\newcommand{\SL}{\mathrm{SL}_2(\C)}
\newcommand{\dd}{\mathrm d}
\newcommand{\ii}{\mathrm i}
\newcommand{\ord}{\operatorname{ord}}
\newcommand{\Sch}[2]{\{#1,#2\}}
\newtheorem{theorem}{Theorem}[section]

\newtheorem{proposition}[theorem]{Proposition}
\newtheorem{corollary}[theorem]{Corollary}
\theoremstyle{definition}
\newtheorem{remark}[theorem]{Remark}

\numberwithin{equation}{section}

\begin{document}
\thispagestyle{plain}
\begin{center}
{\LARGE\bfseries Liouvillian and D-finite obstructions for a\\[3pt]
Klein--Gordon equation with an\\[3pt]
\(\boldsymbol{\alpha}\)-attractor-type potential\par}
\vspace{0.85em}
{\large Benjamin de Zayas Nuñez and Clara Rojas\par}
\vspace{0.3em}
{\small School of Physical Sciences and Nanotechnology, Yachay Tech University,\\
Hda. San José S/N y Proyecto Yachay, 100119 Urcuquí, Ecuador\\
\href{mailto:bdezayas@yachaytech.edu.ec}{bdezayas@yachaytech.edu.ec}\par}
\end{center}
\begin{abstract}
We determine the differential Galois group and D-finiteness obstructions
of the stationary spin-zero Klein--Gordon equation with temporal vector
potential \(V(x)=V_0\exp(a\tanh(bx))\). First, we prove a general
sufficient criterion for a second-order equation over the exponential
extension \(\C(s,e^{as})\), \(s=\tanh(bx)\), to have differential
Galois group \(\SL\) by accounting for movable Riccati poles and
excluding quadratic and finite cases; this criterion applies directly
to a broad class of relativistic mixed scalar--vector and
Schrödinger-type exponential--hyperbolic potentials. For the
\(\alpha\)-attractor-type step, it yields group \(\SL\) and no nonzero
Liouvillian solutions for every \(abV_0\ne0\) and every fixed energy
and mass. Second, because a nonsolvable Galois group does not rule out
special-function solutions, we introduce a realization-deformation
argument---lifting candidate annihilators into \(K[D]\) and varying the
integration constant of the exponential extension to deform the
indicial exponents at a regular singular point---to establish a
coordinate-dependent D-finite obstruction: every nonzero wavefunction
fails to satisfy a linear differential equation of any order with
rational coefficients in each of the three coordinates \(x\), \(s\),
and \(t=e^{as}\). This excludes finite sums, products, derivatives, and
primitives of classical D-finite special functions (including Airy,
Bessel, Whittaker, Gauss hypergeometric, and Heun functions) with
algebraic arguments in any one of those coordinates. The conclusions
carry over to the physical component of the massive, minimally coupled
spin-zero Duffin--Kemmer--Petiau equation.
\end{abstract}

\section{Model and questions}
The hyperbolic-tangent step admits hypergeometric scattering
solutions for the stationary Klein--Gordon equation \cite{Rojas}.
Exponentiating the hyperbolic tangent gives another smooth step, with
constant asymptotic levels on the real line. This profile occurs in
the log-derivative scattering and superradiance study of Salazar
\emph{et al.} \cite{Salazar}. The terminology ``\(\alpha\)-attractor-type''
reflects purely its functional resemblance to exponential
hyperbolic-tangent profiles arising in canonical-field models of
quintessential inflation \cite{Dimopoulos}; the present work does not
treat cosmological field dynamics, as \(x\) here denotes spatial
position, \(b\) is an independent inverse length, and the problem is a
stationary one-dimensional wave equation in a fixed external field.

With \(\hbar=c=1\), let \(V=qA_0\), take the spatial vector potential
to vanish, and write the scalar field as \(e^{-\ii E x^0}\Psi(x)\).
Minimal vector coupling gives
\begin{equation}\label{eq:kg}
 \Psi_{xx}+\big[(E-V_0e^{a\tanh(bx)})^2-m^2\big]\Psi=0 .
\end{equation}
For real parameters and \(b>0\), the left and right limits of \(V\)
are \(V_0e^{-a}\) and \(V_0e^a\). The derivative
\(V_x=abV\,\mathrm{sech}^2(bx)\) fixes the direction of the step,
and \(b^{-1}\) sets its length scale. The asymptotic squared wave
numbers are \(k_\pm^2=(E-V_0e^{\pm a})^2-m^2\). Boundary conditions
select the scattering state. Our conclusions apply to every nonzero
solution, including at a threshold \(k_\pm=0\).

Two distinct questions arise. One asks whether quadratures, algebraic
extensions and exponentials of quadratures suffice to express a
wavefunction. The other asks which special-function expressions can
represent one. Even the Airy equation has a nonsolvable differential
Galois group and named special-function solutions
\cite{SingerIntro,NIST}. We address the second question by specifying
the allowed coordinates, arguments, and operations. Throughout, the
physical parameters are fixed constants. The proofs use complex local
solutions and therefore apply also to real solutions by analytic
continuation at ordinary points.

\paragraph{Main methodological contributions.}
Beyond resolving the specific profile \eqref{eq:kg}, this paper
develops two transferable algebraic--analytic tools for wave equations
with composite exponential--hyperbolic coefficients.
\begin{enumerate}
\item \emph{A general Galois criterion over an exponential extension}
(Proposition~\ref{prop:crit}): rather than applying Kovacic's classical
rational algorithm outside its domain, we prove directly on the
differential field \((\C(s,e^{as}),\partial_s+at\partial_t)\) that any
equation \(D^2y=(c^2t^2+Bt+C_0)y\) with \(B/c\to0\) and \(C_0\to0\) at
\(s=\infty\) has Picard--Vessiot group \(\SL\). As noted in
Remark~\ref{rem:families}, this single criterion immediately covers
wider families of mixed Lorentz scalar--vector Klein--Gordon models and
one-dimensional Schrödinger equations with exponential--hyperbolic
potentials.
\item \emph{A realization-deformation test for multi-coordinate
non-D-finiteness} (Theorem~\ref{th:dfin}): to separate Liouvillian
non-integrability from classical special-function solvability, we lift
candidate rational annihilators in \(s\) or \(t\) to the skew
polynomial ring \(K[D]\), deduce right divisibility by \(L_0\) from
Theorem~\ref{th:physical}, and vary the integration constant of
\(Dt=at\) via \(t_\lambda(s)=\lambda e^{a(s-1)}\). This family of
injective realizations deforms the indicial exponents of \(L_0\) at the
regular singular point \(s=1\) while keeping a fixed operator in \(s\)
and making \(s=1\) ordinary for a generic specialization in \(t\).
\end{enumerate}

\paragraph{Relation to our earlier preprint.}
An earlier preprint by the authors \cite{Preprint26} announced
differential Galois non-integrability and special-function restrictions
for \eqref{eq:kg}. The present manuscript substantially reconstructs
and extends that preliminary treatment: the Galois proof is replaced by
the general exponential-extension criterion of
Proposition~\ref{prop:crit} (explicitly accommodating movable Riccati
poles, resolving the quadratic case via the discriminant, excluding
finite groups via Puiseux valuations compatible with \(D\), and
isolating the sharp nondegeneracy hypothesis \(abV_0\ne0\)); the
unrestricted claim against arbitrary special-function compositions in
\cite{Preprint26} is replaced by the new coordinate-dependent
D-finiteness obstruction of Theorem~\ref{th:dfin} and
Corollary~\ref{cor:class}, complemented by a corrected Schwarzian test
for algebraic changes of variable (Section~\ref{sec:schwarz}); and the
Duffin--Kemmer--Petiau (DKP) reduction is explicitly derived and
delimited to the massive, minimally coupled spin-zero sector
(Section~\ref{sec:normal}).

\section{Normal form and physical scope}\label{sec:normal}
Put \(s=\tanh(bx)\), \(t=e^{as}\), \(F=\C(s)\), and \(K=F(t)\). For
\(a\ne0\), the essential singularity of \(e^{as}\) at \(s=\infty\)
precludes algebraicity over \(F\). Thus the meromorphic realization
of \(K\) is injective. A germ in this realization with zero derivative
is constant, so the constant field is \(\C\). The two derivations are
\begin{equation}\label{eq:d}
 D=\partial_s+at\partial_t,\qquad
 \delta_x=b(1-s^2)D .
\end{equation}
For \(b>0\) real, the real coordinate \(s\) ranges over the open
interval \((-1,1)\). We use the rational function field in \(s\)
for the complex differential equation.

Expanding \(\delta_x^2\) in \eqref{eq:kg} gives
\begin{equation}\label{eq:L0}
 L_0\Psi=0,\qquad L_0=D^2+pD+q_0,\qquad
 p=-\frac{2s}{1-s^2},\quad
 q_0=\frac{(E-V_0t)^2-m^2}{b^2(1-s^2)^2}.
\end{equation}
The substitution \(\Psi=(1-s^2)^{-1/2}y\) gives \(D^2y=Ry\), where
\begin{equation}\label{eq:R}
 R=\tfrac12 Dp+\tfrac14p^2-q_0
   =\frac{m^2-b^2-(E-V_0t)^2}{b^2(1-s^2)^2}.
\end{equation}
In particular, the geometric contribution to \(R\) is
\(-1/(1-s^2)^2\).

\begin{theorem}\label{th:physical}
Fix \(a,b,V_0\in\C^\times\) and \(E,m\in\C\). The Picard--Vessiot
group of \(D^2y=Ry\) over \((K,D)\) is \(\SL\).
Equation~\eqref{eq:kg} has no nonzero Liouvillian solution over
\(\C(\tanh(bx),e^{a\tanh(bx)})\) with derivation \(\dd/\dd x\).
\end{theorem}
Section~\ref{sec:proof} proves the theorem. It includes \(E=0\),
\(m=0\), and \(m^2=b^2\). If \(a=0\), \(b=0\), or \(V_0=0\), the
potential instead equals a constant \(V_c\). Setting
\(k^2=(E-V_c)^2-m^2\) gives the basis \(e^{\ii kx},e^{-\ii kx}\)
for \(k\ne0\), and \(1,x\) for \(k=0\).

\paragraph{Spin-zero DKP sector.}
For \(m\ne0\), let \(\nabla_\mu=\partial_\mu+\ii q A_\mu\) and take
metric signature \((+,-,-,-)\). For the massive minimally coupled
DKP field \(\psi\), use the spin-zero projector \(P\) and define
\(\phi=P\psi\), \(\phi^\mu=P\beta^\mu\psi\). The identity
\(P\beta^\mu\beta^\nu=g^{\mu\nu}P\) projects
\((\ii\beta^\mu\nabla_\mu-m)\psi=0\) to
\[
 m\phi=\ii\nabla_\mu\phi^\mu,\qquad
 m\phi^\nu=\ii\nabla^\nu\phi.
\]
Eliminating \(\phi^\mu\) yields
\((\nabla_\mu\nabla^\mu+m^2)\phi=0\) \cite{Castro}.
With dependence on one spatial coordinate, the stationary physical
component therefore obeys \eqref{eq:kg}. The auxiliary components
follow by covariant differentiation. The massless value in
Theorem~\ref{th:physical} belongs to the Klein--Gordon statement.

\section{The differential Galois group}\label{sec:proof}
We first establish a general sufficient criterion over the exponential
extension \((K,D)\), applicable to a broad class of second-order
equations beyond \eqref{eq:R}. For a rational function of \(s\) (and
its algebraic extensions), a prime below means the ordinary
\(s\)-derivative.

\begin{proposition}[Criterion over an exponential extension]\label{prop:crit}
Let \(a\in\C^\times\), \(F=\C(s)\), \(K=F(t)\), \(Ds=1\), and
\(Dt=at\). Suppose that \(c\in F^\times\) and \(B,C_0\in F\) satisfy
\begin{equation}\label{eq:crit}
 R=c^2t^2+Bt+C_0,\qquad
 \lim_{s\to\infty} B/c=0,\qquad \lim_{s\to\infty} C_0=0 .
\end{equation}
The differential Galois group of \(D^2y=Ry\) over \(K\) is \(\SL\).
\end{proposition}

The Wronskian is constant, so the group lies in \(\SL\).
A proper algebraic subgroup either fixes a line, permutes a pair of
lines, or is finite. A fixed line gives a logarithmic derivative
\(w=Dy/y\in K\). A permuted pair gives conjugate logarithmic
derivatives of degree at most two. A finite group gives algebraic
solutions \cite{vdPS,Kovacic}. We exclude these possibilities with
the actual derivation on the exponential extension. Kovacic's rational
algorithm has its classical input in \(\C(z)\) with ordinary
derivative \cite{Kovacic}; for its systematic application to the
Galoisian integrability of one-dimensional Schrödinger equations see
Acosta-Humánez, Morales-Ruiz, and Weil \cite{Acosta2011}, and for the
broader problem of Liouvillian coefficients see Singer \cite{Singer91}.

\subsection{Local behavior of the Riccati equation}\label{sec:riccati}
We study
\begin{equation}\label{eq:ric}
 Dw+w^2=R .
\end{equation}
A monic irreducible \(q\in F[t]\) of positive degree \(N\) with
\(q\mid Dq\) in \(F[t]\) satisfies \(Dq=aNq\), by comparing leading
terms in \(t\). Writing \(q=\sum_{j=0}^N q_jt^j\) gives
\(q_j'=a(N-j)q_j\). A nonzero rational function has logarithmic
derivative tending to zero at \(s=\infty\). Thus \(q_j=0\) for
\(j<N\), and irreducibility forces \(q=t\).

Let \(\overline F\) be an algebraic closure of \(F=\C(s)\). Because
\(\operatorname{char}(F)=0\), the derivation \(\dd/\dd s\) extends
uniquely to \(\overline F\) (still denoted by a prime), and
\(D=\partial_s+at\partial_t\) extends uniquely to \(\overline F(t)\)
and to every finite algebraic extension of \(K=F(t)\) \cite{vdPS}.
We normalize valuations on \(\overline F(t)\) by assigning value one to
the local parameter \(u\). At \(t=0\) or \(t=\infty\), take \(u=t\) or
\(u=1/t\), respectively. By the Newton--Puiseux theorem over the
algebraically closed field \(\overline F\) of characteristic zero
\cite{Wall}, any element algebraic over \(\overline F(t)\) embeds locally
into a field of formal Puiseux series \(\overline F((u^{1/e}))\) for
some integer \(e\ge1\), with valuation \(\ord_u\) taking values in
\(\frac1e\mathbb Z\). The unique extension of \(D\) to
\(\overline F((u^{1/e}))\) acts termwise on Puiseux monomials by
\begin{equation}\label{eq:val}
 D(h u^\nu)=(h'\pm a\nu h)\,u^\nu,\qquad h\in\overline F,
 \quad \nu\in\tfrac1e\mathbb Z,
\end{equation}
with the upper sign at \(u=t\) and the lower sign at \(u=1/t\), since
\(Dt=at\) and \(D(t^{-1})=-at^{-1}\). In particular,
\(\ord_u(Df)\ge\ord_u(f)\) for all \(f\in\overline F((u^{1/e}))\).
In \eqref{eq:ric} this excludes a pole at \(t=0\): if \(\ord_t(w)=-p<0\),
then \(\ord_t(Dw)\ge -p > -2p = \ord_t(w^2)\), contradicting
\(\ord_t(R)\ge0\). At \(t=\infty\), \(\ord_{1/t}(R)=-2\); the square
of a pole dominates its derivative, so
\begin{equation}\label{eq:leading}
 w\sim\gamma t,\qquad \gamma=+c\ \text{or}\ -c .
\end{equation}

At a root \(\alpha\in\overline F\) of a monic irreducible
\(q\in F[t]\) with \(q\ne t\), take the local parameter
\(u=t-\alpha\). Since \(Dt=at=a(u+\alpha)\), we have
\(Du=d_0+au\) with constant term \(d_0=a\alpha-\alpha'\).
Differentiating \(q(s,\alpha(s))=0\) with respect to \(s\) and
subtracting from \((Dq)(s,\alpha)\) yields
\((Dq)(s,\alpha)=q_t(s,\alpha)\,d_0\). Because \(q\) is separable in
characteristic zero and \(q\nmid Dq\) in \(F[t]\), both
\(q_t(s,\alpha)\) and \((Dq)(s,\alpha)\) are nonzero, so \(d_0\ne0\).
On a Puiseux monomial \(hu^\nu\in\overline F((u^{1/e}))\),
\[
 D(h u^\nu)=\nu h d_0\,u^{\nu-1}+(h'+a\nu h)\,u^\nu .
\]
If an algebraic branch of \eqref{eq:ric} has a pole \(w\sim k u^{-p}\)
with \(k\in\overline F^\times\) and \(p\in\mathbb Q_{>0}\), then
\(Dw\sim-pkd_0u^{-p-1}\) whereas \(w^2\sim k^2u^{-2p}\). Regularity of
\(R\in F[t]\) at \(u=0\) forces \(-p-1=-2p\) and \(-kd_0+k^2=0\), hence
\(p=1\) and \(k=d_0\). Writing the separable factorization
\(q(s,t)=\prod_{i=1}^d(t-\alpha_i)\) over \(\overline F\), we compute
\begin{equation}\label{eq:logderq}
 \frac{Dq}{q}=\sum_{i=1}^d\frac{D(t-\alpha_i)}{t-\alpha_i}
 =\sum_{i=1}^d\left(a+\frac{a\alpha_i-\alpha_i'}{t-\alpha_i}\right).
\end{equation}
Thus at every root \(\alpha_i\) of \(q\), the principal part of
\(Dq/q\) is precisely \(d_{0,i}(t-\alpha_i)^{-1}\), matching the
entire principal part of any rational Riccati solution \(w\in K\) that
has a pole at \(q\).

\subsection{Rational Riccati solutions}
Suppose \(w\in K=F(t)\) solves \eqref{eq:ric}, and form the monic
square-free product \(H\in F[t]\) of all monic irreducibles \(q\ne t\)
at which \(w\) has a pole (setting \(H=1\) if no such pole exists).
Because \(w\) has no pole at \(t=0\), we have \(t\nmid H\) and hence
\(h_0=H(s,0)\ne0\), so \(DH/H=\sum_{q\mid H} Dq/q\) is also regular at
\(t=0\). At every root \(\alpha\in\overline F\) of a factor \(q\mid H\),
the local balance above and \eqref{eq:logderq} show that \(w\) and
\(DH/H\) have simple poles with the exact same principal part
\((a\alpha-\alpha')(t-\alpha)^{-1}\); their difference \(w-DH/H\) is
therefore regular at \(t=\alpha\). At every other finite point
\(t=\beta\in\overline F\) (including \(\beta=0\)), both \(w\) and
\(DH/H\) are regular. Since \(w-DH/H\in F(t)\) has no finite poles
over \(\overline F\), it belongs to the polynomial ring \(F[t]\).

Furthermore, writing \(N=\deg_t H\ge0\), monicity of \(H\) implies
\(DH/H=aN+O(t^{-1})\) at \(t=\infty\), while \eqref{eq:leading} gives
\(w=\gamma t+O(1)\). Consequently, the polynomial \(w-DH/H\in F[t]\)
has degree one in \(t\), and we obtain
\begin{equation}\label{eq:H}
 w=\gamma t+d+\frac{DH}{H},\qquad
 H\in F[t]\ \text{monic},\quad N=\deg_t H\ge0,\quad h_0=H(s,0)\ne0 ,
\end{equation}
with \(d\in F\). Expanding \eqref{eq:H} at infinity as
\(w=\gamma t+w_0+O(t^{-1})\) with \(w_0=d+aN\) and comparing the
coefficient of \(t\) in \eqref{eq:ric} gives
\begin{equation}\label{eq:wzero}
 w_0=\frac{B}{2\gamma}-\frac{\gamma'}{2\gamma}-\frac a2,\qquad
 d=w_0-aN.
\end{equation}
At \(t=0\) the induced derivation on \(F\) is \(\dd/\dd s\).
Evaluating \eqref{eq:H} and \eqref{eq:ric} at \(t=0\) gives
\[
 r=w(s,0)=w_0-aN+\frac{h_0'}{h_0}\in F,\qquad r'+r^2=C_0 .
\]
By \eqref{eq:crit} and the decay of rational logarithmic derivatives
at \(s=\infty\), we have \(r\to-a(N+\tfrac12)\) as \(s\to\infty\).
Because this limit is finite, \(r'\to0\). Passing to the limit in
\(r'+r^2=C_0\) yields \(a^2(N+\tfrac12)^2=0\), which is impossible for
\(a\ne0\) and \(N\in\mathbb Z_{\ge0}\). This excludes rational
solutions while allowing every possible movable pole.

\subsection{Quadratic Riccati solutions}
Suppose \(w_1,w_2\) are distinct conjugate solutions of \eqref{eq:ric}
of degree two over \(K\). Their sum \(T=w_1+w_2\) and squared
difference (discriminant) \(h=(w_1-w_2)^2\) lie in \(K\) and
\(K^\times\), respectively. Subtracting the two Riccati equations
gives \(D(w_1-w_2)=-(w_1+w_2)(w_1-w_2)\), whence
\begin{equation}\label{eq:disc}
 \frac{D(w_1-w_2)}{w_1-w_2}=-T,\qquad \frac{Dh}{h}=-2T.
\end{equation}

Let \(q\in F[t]\) be any monic irreducible with \(q\ne t\), let
\(\alpha\in\overline F\) be a root of \(q\), and put \(u=t-\alpha\).
Because \(q\) is separable over \(F\), \(\alpha\) is a simple root of
\(q\), so the discrete valuation \(\ord_q(h)\) on \(K=F(t)\) equals
the local order \(n=\ord_u(h)\in\mathbb Z\) in \(\overline F((u))\).
Writing \(h=h_1 u^n(1+O(u))\) with \(h_1\in\overline F^\times\) and
recalling \(Du=d_0+au\) (\(d_0=a\alpha-\alpha'\ne0\)), we obtain
\begin{equation}\label{eq:Dhh}
 \frac{Dh}{h}=\frac{h_1'}{h_1}+n\frac{Du}{u}+O(1)
 =n d_0\,u^{-1}+O(1).
\end{equation}
On the other hand, embedding \(w_1,w_2\) into the Puiseux field
\(\overline F((u^{1/2}))\), Section~\ref{sec:riccati} shows that each
conjugate \(w_i\) is either regular at \(u=0\) (\(\ord_u(w_i)\ge0\))
or has a simple pole of the form \(w_i=d_0 u^{-1}+o(u^{-1})\). Let
\(k\in\{0,1,2\}\) be the number of branches among \(w_1,w_2\) that
have a pole at \(u=0\). Then \(T=w_1+w_2=k d_0 u^{-1}+o(u^{-1})\);
since \(T\in K\) has a Laurent expansion in integer powers of \(u\),
its \(u^{-1}\) coefficient is exactly \(k d_0\). Comparing the
\(u^{-1}\) coefficients in \eqref{eq:disc} via \eqref{eq:Dhh} yields
\(n d_0=-2k d_0\), and because \(d_0\ne0\),
\[
 n=-2k\in\{0,-2,-4\}.
\]
In particular, \(n\le0\), which immediately excludes any positive-order
zero of \(h\) at a transverse divisor \(q\ne t\). (Equivalently, if
\(n>0\), then \(T=-\frac{n}{2}d_0u^{-1}+O(1)\) while
\(\sqrt h=O(u^{n/2})\), forcing both roots
\(w_{1,2}=\frac12(T\pm\sqrt h)\) to have a simple pole with leading
coefficient \(-\frac{n}{4}d_0\ne d_0\), contradicting
Section~\ref{sec:riccati}.) Moreover, since both conjugates have pole
order at most one (\(\ord_u(w_i)\ge-1\)), their difference satisfies
\(\ord_u(w_1-w_2)\ge-1\), whence
\(n=\ord_u((w_1-w_2)^2)=2\ord_u(w_1-w_2)\ge-2\). This excludes
\(n=-4\) (for \(k=2\), the identical leading terms \(d_0 u^{-1}\)
actually cancel in \(w_1-w_2\), forcing \(n>-2\)).

Therefore every monic irreducible \(q\in F[t]\) other than \(t\)
satisfies \(\ord_q(h)\in\{0,-2\}\). Since \(F[t]\) is a unique
factorization domain, collecting all monic irreducibles \(q\ne t\)
with \(\ord_q(h)=-2\) into a monic square-free polynomial
\(H\in F[t]\) gives the factorization
\begin{equation}\label{eq:hshape}
 h=h_*\,\frac{t^j}{H^2},\qquad h_*\in F^\times,\quad j\in\mathbb Z,
 \quad H\in F[t]\ \text{monic},\quad H(s,0)\ne0 .
\end{equation}
At \(t=\infty\) the two leading signs in \eqref{eq:leading} must be
opposite. Equal signs would give \(T\sim\pm2ct\) a pole of order one
at \(u=1/t\), whereas \eqref{eq:val} implies
\(\ord_{1/t}(D(w_1-w_2)/(w_1-w_2))\ge0\), contradicting
\eqref{eq:disc}. Hence \(w_1-w_2\sim\pm2ct\) and \(h\sim4c^2t^2\).
Comparing the leading behavior at \(t=\infty\) with \eqref{eq:hshape}
gives \(h_*=4c^2\) and \(j-2\deg_t H=2\), so
\[
 h=\left(\frac{2c\,t^{\deg_t H+1}}{H}\right)^2\ \text{in }K .
\]
Thus \(\sqrt h\in K\), and both roots \((T\pm\sqrt h)/2\) would belong
to \(K\). This contradicts their assumed degree two and excludes the
imprimitive case.

\subsection{Finite groups, the physical equation, and broader families}
If \(y\ne0\) is algebraic over \(K\), \eqref{eq:val} at infinity gives
\(\ord_{1/t}(D^2y)\ge\ord_{1/t}(y)\), whereas
\(\ord_{1/t}(Ry)=\ord_{1/t}(y)-2\). Thus no such \(y\) solves the
linear equation. A finite Picard--Vessiot group would make every
solution algebraic. This last exclusion proves
Proposition~\ref{prop:crit}.

For \eqref{eq:R}, the three functions in the criterion are
\begin{equation}\label{eq:abc}
 c=\frac{\ii V_0}{b(1-s^2)},\qquad
 B=\frac{2EV_0}{b^2(1-s^2)^2},\qquad
 C_0=\frac{m^2-b^2-E^2}{b^2(1-s^2)^2}.
\end{equation}
Here \(c\ne0\), \(B/c=O(s^{-2})\) and \(C_0=O(s^{-4})\).
The full group follows at every fixed \(E,m\).
Since \(\SL\) is connected and nonsolvable, its Picard--Vessiot
extension is not Liouvillian \cite{vdPS,SingerIntro}. A single nonzero
Liouvillian solution would yield an independent second one by
\(y\int y^{-2}\,\dd s\). Multiplication by \((1-s^2)^{-1/2}\) is an
algebraic gauge. Rescaling the derivation by \(b(1-s^2)\in K^\times\)
preserves Liouvillian towers, since both ordinary and logarithmic
derivatives rescale by the same factor. This proves
Theorem~\ref{th:physical}. Adjoining the coordinate
\(x=b^{-1}\operatorname{artanh}(s)\) also preserves the
nonexistence conclusion, since \(Dx=1/[b(1-s^2)]\in K\).

\begin{remark}[Broader families of potentials]\label{rem:families}
Proposition~\ref{prop:crit} yields the same \(\SL\) Galois group and
Liouvillian non-integrability, without altering the proof, for two
wider classes of stationary wave equations with \(a,b\in\C^\times\):
\begin{itemize}
\item \emph{Mixed Lorentz scalar--vector and hyperbolic-tangent
couplings.} Consider
\(\Psi_{xx}+[(E-V(x))^2-(m+S(x))^2]\Psi=0\) with
\(V(x)=V_0 e^{a\tanh(bx)}+V_1\tanh(bx)\) and
\(S(x)=S_0 e^{a\tanh(bx)}+S_1\tanh(bx)\). Whenever \(S_0^2\ne V_0^2\),
the reduced coefficient \eqref{eq:R} takes the form \eqref{eq:crit}
with \(c=\sqrt{S_0^2-V_0^2}/[b(1-s^2)]\in F^\times\),
\(B=2[EV_0+mS_0+(S_0S_1-V_0V_1)s]/[b^2(1-s^2)^2]\), and
\(C_0=O(s^{-2})\), so \(B/c=O(s^{-1})\to0\) and \(C_0\to0\) at
\(s=\infty\) for every fixed \(E,m\in\C\).
\item \emph{Stationary one-dimensional Schrödinger equations} (cf.\ \cite{Acosta2011}). Consider
\(-\Psi_{xx}+[U_2 e^{2a\tanh(bx)}+U_1 e^{a\tanh(bx)}
+U_0(\tanh(bx))]\Psi=E\Psi\), where \(U_0(s)\in\C(s)\) satisfies
\(U_0(s)=o(s^4)\) as \(s\to\infty\) (including any polynomial of degree
at most three in \(\tanh(bx)\), such as a hyperbolic-tangent step or a
Pöschl--Teller \(\mathrm{sech}^2(bx)=1-s^2\) perturbation). If
\(U_2\ne0\), Proposition~\ref{prop:crit} applies directly over
\(K=\C(s,e^{as})\) with \(c=\sqrt{U_2}/[b(1-s^2)]\). If \(U_2=0\) and
\(U_1\ne0\) (a single exponential step), setting \(\tilde t=e^{(a/2)s}\)
writes \(e^{as}=\tilde t^2\) and yields differential Galois group
\(\SL\) over \(\tilde K=\C(s,\tilde t)\supset K\); because
\(G_{\tilde K}\subseteq G_K\subseteq\SL\), the group over \(K\) is
likewise \(\SL\), ruling out nonzero Liouvillian solutions at every
fixed energy \(E\in\C\).
\end{itemize}
\end{remark}

\section{Special functions and the choice of coordinate}\label{sec:dfin}
A local germ \(f(q)\) is \emph{D-finite in \(q\)} if a nonzero linear
ordinary differential operator of finite order in \(\dd/\dd q\) with
coefficients in \(\C(q)\) annihilates it. Equivalently, its derivatives
span a finite-dimensional space over \(\C(q)\) \cite{Stanley}.
We apply this definition to the wavefunction \(\Psi\) itself.

While the spatial coordinate \(x\) is ruled out directly by infinitely
many essential singularities, our proof for \(s\) and \(t\) relies on a
general algebraic--analytic mechanism: any rational annihilator lifts
to a nonzero operator \(L\in K[D]\) that must be divisible on the right
by \(L_0\), and deforming the integration constant of \(Dt=at\) via
\(t_\lambda(s)=\lambda e^{a(s-1)}\) continuously varies the indicial
exponents of \(L_0\) at the regular singular point \(s=1\) while
leaving a candidate operator in \(s\) unchanged and turning \(s=1\)
into an ordinary point for a generic specialization in \(t\).

\begin{theorem}\label{th:dfin}
Under the hypotheses of Theorem~\ref{th:physical}, every nonzero local
solution of \eqref{eq:kg} fails to be D-finite in each of the three
coordinates \(x\), \(s=\tanh(bx)\), and \(t=e^{as}\), on locally
invertible branches.
\end{theorem}

\begin{proof}
\emph{The coordinate \(x\).}
At each \(x_n=\ii\pi(n+\tfrac12)/b\), \(n\in\mathbb Z\),
\(\tanh(bx)\) has a simple pole and \(e^{a\tanh(bx)}\) has an
essential singularity. So does \(m^2-(E-V_0e^{a\tanh(bx)})^2\).
For instance, one may approach the pole on paths where the exponential
tends to zero and on paths where its modulus tends to infinity.
A linear equation over \(\C(x)\) has only finitely many finite
singular points. Pick \(x_n\) ordinary for a proposed annihilator,
and continue the germ along a path avoiding the singularities of both
equations. Every solution of the proposed annihilator extends
holomorphically across \(x_n\). Its quotient \(\Psi_{xx}/\Psi\) would
be meromorphic there, even at a zero of finite order. Equation
\eqref{eq:kg} contradicts this.

\emph{Lifting an annihilator in \(s\) or \(t\) and right division in \(K[D]\).}
An annihilator in \(s\) yields directly a nonzero operator
\(L\in\C(s)[D]\). For an annihilator in \(t\), we account explicitly
for the noncommutativity between \(D\) and powers of \(t\). On the
curve \(t=e^{as}\), where \(Dt=at\), the first two derivatives are
\begin{equation}\label{eq:dtlift}
 \frac{\dd}{\dd t}=\frac{1}{at}D,\qquad
 \frac{\dd^2}{\dd t^2}=\frac{1}{at}D\!\left(\frac{1}{at}D\right)
 =\frac{D^2-aD}{a^2t^2}.
\end{equation}
Because \(D(t^{-k}f)=t^{-k}(D-ka)f\), induction gives
\[
 \frac{\dd^k}{\dd t^k}=\frac{1}{(at)^k}(D-(k-1)a)\cdots(D-a)D
 \in\C(t)[D]
\]
for every integer \(k\ge1\), with leading term \((at)^{-k}D^k\).
Hence any nonzero differential operator of order \(n\) in
\(\C(t)[\dd/\dd t]\) lifts to a nonzero operator
\(L\in\C(t)[D]\) of the same order \(n\).

In either case (\(L\in\C(s)[D]\) or \(L\in\C(t)[D]\)), we view \(L\)
inside the skew polynomial ring \(K[D]\) over the differential field
\((K,D)\), governed by the commutation law \(D\cdot f=f\cdot D+(Df)\)
for \(f\in K\). Because \(K[D]\) is a right Euclidean domain
\cite{vdPS} and \(L_0=D^2+pD+q_0\in K[D]\) is monic of order two,
right division produces unique operators \(M,U\in K[D]\) such that
\[
 L=M L_0+U,
\]
with either \(U=0\) or \(\ord U<2\). Applying this identity to the
nonzero local solution \(\Psi\) (which satisfies \(L\Psi=0\) and
\(L_0\Psi=0\)) yields \(U\Psi=0\). A nonzero remainder of order zero
(\(U\in K^\times\)) is impossible. A remainder of order one,
\(U=u_1D+u_0\) with \(u_1\in K^\times\), would imply
\(D\Psi/\Psi=-u_0/u_1\in K\), so that
\(w=Dy/y=D\Psi/\Psi-s/(1-s^2)\in K\) would be a rational Riccati
solution for \(D^2y=Ry\), contrary to Theorem~\ref{th:physical}.
Therefore \(U=0\) and \(L=M L_0\) holds as an exact operator identity
in \(K[D]\).

\emph{Variation of the exponential parameter.}
Because \(L=M L_0\) is an algebraic identity in \(K[D]\), it is
preserved under any differential realization of \((K,D)\) that shifts
the integration constant of \(Dt=at\). Specifically, for every
\(\lambda\in\C^\times\), the function
\[
 t_\lambda(s)=\lambda e^{a(s-1)}
\]
satisfies \(t_\lambda'(s)=at_\lambda(s)\) and is transcendental over
\(\C(s)\) (as \(a\ne0\) creates an essential singularity at
\(s=\infty\)). Substituting \(t\mapsto t_\lambda(s)\) and
\(D\mapsto\dd/\dd s\) therefore defines an injective differential
homomorphism from \((K,D)\) into the field of meromorphic functions on
\(\C\), turning \(L=M L_0\) into an identity of analytic differential
operators \(L^{(\lambda)}=M^{(\lambda)}L_0^{(\lambda)}\).
Consequently, on any slit neighborhood of \(s=1\) (avoiding the
isolated poles of the coefficients of \(M^{(\lambda)}\), and extending
by analytic continuation), \emph{every} local solution of
\(L_0^{(\lambda)}\Psi=0\) also solves \(L^{(\lambda)}\Psi=0\).

At \(s=1\), putting \(z=s-1\), the specialized operator
\(L_0^{(\lambda)}\) has a regular singular point: since
\(p=-2s/(1-s^2)=2(1+z)/[z(2+z)]\) and \(t_\lambda(1)=\lambda\), we
have \(zp\to1\) and
\(z^2q_0^{(\lambda)}\to((E-V_0\lambda)^2-m^2)/(4b^2)\) as \(z\to0\).
Its indicial equation \(\rho(\rho-1)+\rho\lim_{z\to0}(zp)
+\lim_{z\to0}(z^2q_0^{(\lambda)})=0\) simplifies to
\begin{equation}\label{eq:indicial}
 \rho^2+\frac{(E-V_0\lambda)^2-m^2}{4b^2}=0 .
\end{equation}
Because \(V_0\ne0\), as \(\lambda\) varies over \(\C^\times\) the
constant term \(((E-V_0\lambda)^2-m^2)/(4b^2)\) assumes all complex
values except possibly \((E^2-m^2)/(4b^2)\), so the exponents \(\rho\)
sweep out an uncountable subset of \(\C\). Excluding the countable
resonant set \(2\rho\in\mathbb Z\) still leaves uncountably (in
particular, infinitely) many nonresonant exponents. For each such
nonresonant \(\rho\), Frobenius theory \cite{NIST} yields a nonzero
local solution \(\Psi_\lambda(z)=z^\rho(1+O(z))\) on a slit disk
around \(s=1\).

\emph{The coordinate \(s\).}
If \(L=\sum_j f_j(s)D^j\in\C(s)[D]\), its coefficients do not depend
on \(t\), so its specialization \(L^{(\lambda)}=L\) is the exact same
fixed operator for every \(\lambda\in\C^\times\). Write each nonzero
rational coefficient as \(f_j(s)=z^{v_j}(d_j+O(z))\) with
\(d_j\ne0\), and let \(\ell=\min_j(v_j-j)\). Substituting the
Frobenius solution \(\Psi_\lambda(z)=z^\rho(1+O(z))\) into
\(L\Psi_\lambda=0\) and extracting the coefficient of the lowest power
\(z^{\rho+\ell}\) forces
\[
 I_L(\rho)=\sum_{v_j-j=\ell}d_j\,\rho(\rho-1)\cdots(\rho-j+1)=0 ,
\]
where the empty product for \(j=0\) equals one. Because the summands
correspond to distinct indices \(j\) and thus have distinct degrees in
\(\rho\), \(I_L(\rho)\) is a fixed nonzero polynomial. It can have
only finitely many roots, contradicting the infinitely many
nonresonant exponents \(\rho\) from \eqref{eq:indicial}.

\emph{The coordinate \(t\).}
Suppose now that \(L=\sum_{j=0}^n g_j(t)D^j\in\C(t)[D]\) with
\(g_n\not\equiv0\). The poles of the rational functions
\(g_0,\dots,g_n\) together with the zeros of the leading coefficient
\(g_n\) form a finite subset \(\Sigma\subset\C\). Because
\(t_\lambda(1)=\lambda\), any choice \(\lambda\in\C^\times\setminus\Sigma\)
ensures that every specialized coefficient \(g_j(t_\lambda(s))\) is
holomorphic in a neighborhood of \(s=1\) and that the leading
coefficient satisfies \(g_n(t_\lambda(1))=g_n(\lambda)\ne0\).
Dividing \(L^{(\lambda)}\) by \(g_n(t_\lambda(s))\) puts it in monic
form with coefficients holomorphic at \(s=1\); hence \(s=1\) is an
ordinary point of \(L^{(\lambda)}\), and every local solution of
\(L^{(\lambda)}\Psi=0\) near \(s=1\) extends holomorphically across
\(s=1\). Since the nonresonance condition \(2\rho\notin\mathbb Z\) in
\eqref{eq:indicial} excludes at most countably many values of
\(\lambda\), there exist uncountably many
\(\lambda\in\C^\times\setminus\Sigma\) with \(2\rho\notin\mathbb Z\)
(and in particular \(\rho\notin\mathbb Z\)). For such a \(\lambda\),
the Frobenius solution \(\Psi_\lambda(z)=z^\rho(1+O(z))\) of
\(L_0^{(\lambda)}\Psi=0\) has a branch point at \(s=1\) and cannot be
holomorphic there, contradicting \(L^{(\lambda)}\Psi_\lambda=0\).
\end{proof}

\begin{corollary}\label{cor:class}
Fix one coordinate \(q\in\{x,s,t\}\). A nonzero wavefunction cannot
be formed by finite sums, products, derivatives or primitives in \(q\)
of D-finite functions with arguments algebraic over \(\C(q)\) and
algebraic prefactors. This includes expressions formed in this way
from Airy, Bessel, Whittaker, Gauss hypergeometric and Heun functions
with fixed parameters.
\end{corollary}
\begin{proof}
These families satisfy linear equations with rational coefficients
in their standard arguments \cite{NIST} (for Heun-type equations and
their special solutions in quantum and field-theoretic wave models, see
also \cite{Karayer2015,Birkandan2017}). Sums and products preserve
finite derivative spans by linearity and the product rule. A primitive
enlarges such a span by at most one generator. For an algebraic
argument \(\xi(q)\), the chain rule gives a finite derivative span
over the finite extension \(\C(q,\xi)\), which is stable under
\(\dd/\dd q\). Restriction of scalars gives a finite span over
\(\C(q)\) \cite{Stanley}. Theorem~\ref{th:dfin} now applies.
\end{proof}

\subsection{Changes that preserve the reduced equation}\label{sec:schwarz}
There is also a short direct test for algebraic changes of the
independent variable that preserve the normal (reduced) form. In a
coordinate \(q\), let \(Y_{qq}=R_qY\) be the reduced equation. Under
a locally invertible change \(\zeta=\xi(q)\) (\(\xi'\ne0\)), the chain
rule introduces a first-derivative term \((\xi''/\xi')U_\zeta\).
Applying the Liouville gauge \(Y=(\xi')^{-1/2}U(\zeta)\) removes this
first derivative at the expense of adding half the Schwarzian
derivative to \(R_q\) before rescaling by \((\xi')^{-2}\), yielding
\cite{NIST}
\begin{equation}\label{eq:sch}
 U_{\zeta\zeta}=\widetilde R U,\qquad
 \widetilde R(\xi(q))=\frac{R_q+\tfrac12\Sch{\xi}{q}}{(\xi')^2},
 \qquad \Sch{\xi}{q}=\frac{\xi'''}{\xi'}-\frac32
                  \left(\frac{\xi''}{\xi'}\right)^2 .
\end{equation}
Because the algebraic closure \(\overline{\C(q)}\) is closed under
\(\dd/\dd q\), if \(\xi\) is algebraic over \(\C(q)\), then
\(\xi'\), \(\xi''\), \(\xi'''\), and \(\Sch{\xi}{q}\) are algebraic
over \(\C(q)\) as well. Requiring a rational reduced coefficient
\(\widetilde R(\zeta)\in\C(\zeta)\) would then force
\[
 R_q=(\xi')^2\widetilde R(\xi(q))-\tfrac12\Sch{\xi}{q}
\]
to be algebraic over \(\C(q)\).

We test this condition for each of the three coordinates:
for \(q=x\), the reduced coefficient
\(R_x=m^2-(E-V_0e^{a\tanh(bx)})^2\) has essential singularities at
\(x_n=\ii\pi(n+\tfrac12)/b\) and cannot be algebraic over \(\C(x)\);
for \(q=s\), \eqref{eq:R} is a quadratic polynomial in the
transcendental element \(t=e^{as}\) with nonzero leading coefficient
\(-V_0^2/[b^2(1-s^2)^2]\), so \(R_s\) is transcendental over
\(\C(s)\); and for \(q=t\), applying \eqref{eq:sch} to \(t=e^{as}\)
with \(\Sch{t}{s}=-a^2/2\) and \(t_s=at\) gives the reduced
coefficient
\begin{equation}\label{eq:Rt}
 R_t=\frac{m^2-b^2-(E-V_0t)^2}
            {a^2b^2t^2(1-s^2)^2}-\frac{1}{4t^2}.
\end{equation}
Here \(s\) is transcendental over \(\C(t)\) by algebraic independence
of \(s\) and \(e^{as}\), and the numerator \(m^2-b^2-(E-V_0t)^2\) is
not identically zero because \(V_0\ne0\). Algebraicity of \(R_t\) over
\(\C(t)\) would make \((1-s^2)^2\), and hence \(s\), algebraic over
\(\C(t)\), a contradiction. Thus none of the three coordinates
\(x,s,t\) admits an algebraic change of variable (accompanied by the
Liouville gauge) to a rational reduced equation. In particular, the
\(q=t\) case rules out all transformations for which \(t\) is a
nonconstant rational function of \(\zeta\).

\section{Mathematical and physical consequences}
The two obstructions answer genuinely different questions about the
same wavefunction. To see why a separate special-function proof is
indispensable, consider \(D^2y=a^2t^2y\), which satisfies
Proposition~\ref{prop:crit} (with \(c=a\), \(B=C_0=0\)) and therefore
has differential Galois group \(\SL\) and no nonzero Liouvillian
solution over \(K\). Nevertheless, changing to \(t\) gives the
modified Bessel equation \(t^2y_{tt}+ty_t-t^2y=0\), whose solutions
are the classical D-finite special functions \(I_0(t)\) and \(K_0(t)\)
\cite{NIST}. The exponential hyperbolic-tangent step in \eqref{eq:kg}
exhibits the stronger, coordinate-dependent obstruction of
Theorem~\ref{th:dfin}, ruling out all individual solutions and finite
combinations of the form specified in Corollary~\ref{cor:class} in
each of the coordinates \(x\), \(s\), and \(t\).

From the standpoint of relativistic scattering physics, this result
clarifies a sharp structural contrast between two smooth step
potentials. Whereas the ordinary hyperbolic-tangent step
\(V(x)=V_0\tanh(bx)\) reduces to the Gauss hypergeometric equation and
admits closed-form scattering wavefunctions and reflection/transmission
coefficients \cite{Rojas}, exponentiating the step breaks both
Liouvillian solvability and coordinate-wise D-finiteness across all
fixed values of energy and mass---including the Klein-paradox and
superradiant energy regimes bounded by the asymptotic thresholds
\(k_\pm^2=(E-V_0e^{\pm a})^2-m^2\). Consequently, in contrast to relativistic wave equations solvable in
terms of hypergeometric or Heun-type functions \cite{Rojas,Birkandan2017},
exact connection formulas and scattering amplitudes for \eqref{eq:kg}
cannot be assembled from classical D-finite special functions
(hypergeometric, Whittaker, Bessel, Airy, or Heun) evaluated at
algebraic arguments in \(x\), \(s=\tanh(bx)\), or
\(t=e^{a\tanh(bx)}\). This provides a
theoretical justification for resorting to numerical propagation
schemes, such as the log-derivative method applied to superradiance
in this potential by Salazar \emph{et al.} \cite{Salazar}.

At the same time, the scope of these obstructions should be stated
with precision. First, expressions involving unrestricted
transcendental arguments or nested special functions outside the
classes of Corollary~\ref{cor:class} remain mathematically open.
Second, non-Liouvillianity and non-D-finiteness do not impede the
existence of smooth local analytic solutions on the physical real
line, convergent series expansions, semiclassical approximations, or
accurate numerical scattering calculations for all regular initial
data. Finally, although the profile \(V_0\exp(a\tanh(bx))\) takes its
name from its functional similarity to exponential hyperbolic-tangent
potentials in \(\alpha\)-attractor quintessential inflation
\cite{Dimopoulos}, our obstructions pertain strictly to the linear
stationary relativistic wave equation \eqref{eq:kg} (and the
equivalent massive spin-zero DKP sector) in a fixed external
background, rather than to Hamiltonian integrability or dynamical
chaos in cosmological field theories.

From a broader methodological viewpoint, the value of the present work
extends well beyond the single profile \eqref{eq:kg}. Pairing the
general exponential-extension Galois criterion of
Proposition~\ref{prop:crit} (together with Remark~\ref{rem:families})
with the realization-deformation technique \(t_\lambda(s)\) of
Theorem~\ref{th:dfin} provides a reusable template for rigorously
separating Liouvillian non-integrability from D-finite special-function
obstructions across relativistic and non-relativistic wave equations
with composite exponential--hyperbolic potentials.

\paragraph{Acknowledgments and declarations.}
Yachay Tech University supported this work as part of the authors'
regular academic and research activities. The authors report no
external grant or conflict of interest. This theoretical study involved
no human or animal subjects and generated no experimental data.

\paragraph{Declaration of generative AI and AI-assisted technologies in the writing process.}
During the preparation of this work, the authors used ChatGPT (OpenAI)
to assist with language editing and the preparation of the revised
text. After using this tool, the authors reviewed and edited the
content as needed and take full responsibility for the content of the
publication.

\end{document}